\documentclass[12pt]{article}
\usepackage{datetime}
\newdateformat{monthyeardate}{
  \monthname[\THEMONTH], \THEYEAR}
  
\usepackage{algorithm}
\usepackage{algpseudocode}
\usepackage{pifont}
\usepackage{latexsym}
\usepackage{amsmath}
\usepackage{amsfonts}
\usepackage{amssymb}
\usepackage{amsthm}
\usepackage{setspace}
\usepackage{graphicx}
\usepackage{caption}
\usepackage{subcaption}
\usepackage[titletoc,title]{appendix}

\usepackage{multirow, multicol}
\usepackage{epsfig, subfloat}
\usepackage{pstricks, sgamevar, egameps}
\usepackage{hyperref}
\usepackage{url}
\usepackage{pgfplots}
\usepackage{tikz}
\usetikzlibrary{positioning,chains,fit,shapes,calc}
\usepackage{tkz-graph}

\usepackage{enumerate}
\usepackage{natbib} 

\usepackage{geometry}
\newgeometry{margin=1in}

\newtheorem{theorem}{Theorem}[section]
\newtheorem{lemma}[theorem]{Lemma}

\newtheorem{assumption}[theorem]{Assumption}
\theoremstyle{remark}

\numberwithin{equation}{section}

\theoremstyle{remark}

\begin{document}

\title{\bf The Instrumental Variable Method for Estimating Local Average Treatment Regime Effects\thanks{We are grateful to Han Hong for valuable suggestions. We also thank Yiming He and Dean Eckles for helpful comments. All errors are ours.}}
\author{
  Thai T. Pham\footnote{Graduate School of Business, Stanford University. Email: \texttt{thaipham@stanford.edu}}
  \and
  Weixin Chen\footnote{Department of Economics, Stanford University. Email: \texttt{weixinc@stanford.edu}}
}
\date{\monthyeardate\today}
\maketitle

\sloppy 

\bigskip
\begin{abstract}
We propose the instrumental variable regime (IVR) method to estimate the causal effects of multiple sequential treatments. This method serves to address the problem of endogenous selections of sequential treatments. An IVR is a sequence of instrumental variables in which each IV instruments for an endogenous treatment variable. Our proposed method generalizes the LATE model in \citet{IA1994} from a single treatment to many treatments applied sequentially. More precisely, with the IVR this model allows for estimating the local average treatment regime effects (LATRE), possibly conditional on a set of initial covariates. Though there exist studies in this area that use IVR, all of them require a structural functional form assumption. Our method is novel in that we do not require any such assumption. Thus unlike previous approaches, ours is robust to model misspecifications, which usually occur in treatment regime settings. The ideas and estimators in this paper are motivated and illustrated through a contextual example showing the use of IVR in estimating the treatment regime effect of advertisements on purchasing behaviors when advertisements are displayed in multiple periods. We demonstrate the performance of the proposed method with simulations. 
\end{abstract}

\noindent%
{\it Keywords:}  instrumental variable regime, treatment regime, endogenous selection, non-structural model, causal effect, advertisement effect.
\vfill

\section{Introduction}
\label{sec:intro}

Estimating treatment effects has been a major concern among researchers across different fields including economics, statistics, epidemiology, sociology, etc. Within this research area, estimating treatment regime effects (i.e., effects of a set of sequential treatments over time) has drawn a lot of attention from scientists and practitioners. In this paper, we contribute to this large literature by developing an instrumental variable regime (IVR) framework to estimate the effect of a treatment sequence on a pre-defined outcome when the treatment selections are endogenous. The IVR is a set of instrumental variables each of which instruments for an endogenous treatment variable. Technically, we generalize the idea of the LATE model in \citet{IA1994} from a single treatment to including many sequential ones. Our proposed model also allows for the estimation of a form of local average treatment effect of a treatment regime. More similarly to \citet{A2003}, we allow conditioning on a set of initial covariates of the treatment effect. 

The endogeneity problem usually occurs in observational and even in experimental studies: individuals usually make decisions about treatments based on the evaluation of future potential outcomes. It is important to resolve this issue because otherwise, estimations of causal effects will be affected by a part generated by the selection process rather than only the true treatment effects. This mis-estimation invalidates all causal inference decisions. 

There has been a lot of research on the use of IV in studying causal effects, some of which address the endogeneity problem in the sequential treatment (or longitudinal) setting (see, e.g., \cite{HL2004, O2012, SHBT2012, W2002}). However, to the best of our knowledge all of these precedent methods require a structural functional form assumption on the potential outcomes. In the treatment regime setting, this assumption can hardly be valid and thus, these methods can easily suffer from model misspecifications. In contrast, our proposed method does not make any such assumption and is therefore free from the model misspecification issue.  

On the other note, \cite{HN2007} attack the same question with some established theoretical results. However, their methodology is too complex and not so practical. Our proposed method, on the other hand, is reasonably simple to use in applied work. 

We present our ideas in the context of evaluating the effect of a sequence of advertisements (ads) on Facebook users' purchasing behaviors over time. Ideally, we are interested in the treatment regime effect of users reading ads on their purchasing decisions. However, whether an user reads the ads is unknown. Therefore, we answer a proxy question: what is the treatment regime effect of users clicking ads on their purchasing decisions?

More precisely, we focus on one product and one ad in each time period. At time $T = 0$ users observe their initial information $X_0$. Then they make decisions about whether to click the ad ($W_0 = 1$) or not ($W_0 = 0$). At time $T \geq 1$, users observe the outcome $Y_T$ and other information $X_T$. Then they make decisions about whether to click the ad or not ($W_1 \in \{0, 1\}$). We can think of $X_i$ as a set of personal information given on Facebook. The outcome $Y_T$ is the purchasing decision of the product at time $T$.

The endogenous treatment selection occurs as users are more likely to click the ad if they are more likely to purchase the product, thus biasing their clicking decisions. Hence in this setting the sequential randomization assumption made in traditional treatment regime literature (see e.g., \citet{ORR2010, HL2004, P2016, W2002}) fails to be true; the identification results are therefore no longer valid. 

Our proposed solution is to introduce a set of instrumental variables or an IVR $Z_i$'s where each $Z_i$ instruments for $W_i$. In this setting, $Z_i$ is the binary variable indicating whether the ad is displayed on user $i$'s Facebook homepage.

As mentioned earlier, there is a literature on how to use IV's to estimate treatment regime effects (see, e.g., \cite{HL2004, W2002}). However, all of them require a (usually linear) functional form assumption on the potential outcomes, and using IV's in the same way as we deal with linear regression models. Unlike them, we use IVs combined with propensity scores in a way that does not need to make such a functional form assumption; thus, our method is robust to (outcome) model misspecifications. 

With the introduction of the IVR $Z_i$'s, we can estimate the local average treatment regime effects (LATRE) for different compliance types, possibly conditional on a set of initial covariates. First, a complier refers to a user who clicks the ad if it is displayed while a non-complier does not click when the ad is displayed; there are many compliance types in a multi-period setting: a user can be a complier in some periods but a non-complier in others. Second, the set of initial covariates is useful. For example in our setting, we can condition on the gender variable; then, we can estimate the LATRE for female users and for male users separately. 

The rest of the paper is organized as follows. Section \ref{IVRmodel} describes the instrumental variable regime model. Section \ref{illusModel} discusses the identification results for the two-period model. Section \ref{mainIden} gives main identification results about the estimation of the local average treatment regime effect, perhaps conditional on a set of initial covariates. Section \ref{simu} focuses on the application of the model. Specifically, we run the proposed framework on a two-period setting using a simulated data set. Section \ref{concl} reviews related literature and concludes.

\section{The Instrumental Variable Regime Model}
\label{IVRmodel}

The Instrumental Variable Regime (IVR) model is built upon the Treatment Regime model (see e.g., \citet{CM2013, MVR2001, ORR2010, P2016}) and the LATE model (see \citet{IA1994, A2003, AIR1996}). This model addresses the endogeneity problem in the selections of treatments when the treatments are applied sequentially in multiple periods. The model is specified as follows.

There are $N$ users and $(T + 1)$ periods where $T \geq 1$. Each user $i$'s data is comprised of observations
\begin{equation*} 
	(X_0^i, W_0^i, Y_1^i, X_1^i, W_1^i, ..., Y_T^i, X_T^i, W_T^i, Y_{T + 1}^i, X_{T + 1}^i).
\end{equation*}
Here, $X_0^i$ is the initial set of covariates; for $j \in \{1, ..., T + 1\}$, $X_j^i$ is the set of covariates in period $j$ after receiving treatment $W_{j - 1}^i$ but before receiving treatment $W_j^i$. For $j \in \{0, 1, ..., T\}$, $W_j^i$ is the treatment whose value, that is, treatment level in period $j$ belongs to the set $\mathcal{W}_j^i$. In this setting, we assume $\mathcal{W}_j^i = \{0, 1\}$ for each $j$. After making a treatment decision in period $j \in \{0, 1, ..., T\}$, there is an observed outcome in period $(j + 1)$ denoted by $Y_{j + 1}^i$. We assume that $X_j^i$'s are exogenous and $Y_j^i$'s are functions of only user $i$'s past information. That is, we assume away the interference effects. These effects exist in many settings, but not in ours so long as the purchasing decision of each user is not publicly observed. In the latter parts of this paper, we will suppress the index $i$ when there is no ambiguity.

We proceed by defining
\begin{equation*} 
	W_j = 0 \text{ for } j < 0;\; X_j = 0 \text{ for } j < 0;\; Y_j = 0 \text{ for } j < 1.
\end{equation*}
Also, 
\begin{equation*}
	O_0 = X_0; \text{ and } O_j = (Y_j, X_j) \text{ for all } j \in \{1, 2, ..., T + 1\}.
\end{equation*}
We also use overbars with a subscript $j$ to denote the present variable at time $j$ and all its past values. For example, $\overline{O}_j = (O_0, O_1, ..., O_j)$. We use notations with no subscript to denote the whole history. For example, $O = \overline{O}_{T + 1}$. Moreover, as discussed above $Y_{j+1}$ is a function of $\overline{O}_j$ and $\overline{W}_j$; $W_j$ is a function of $\overline{O}_j$; we suppress the dependency on $\overline{O}_j$ in the latter discussion.

Among many, \citet{ORR2010, P2016} make the sequential randomization assumption that each $W_j$ is independent of all potential outcomes given past information $(\overline{O}_j, \overline{W}_{j - 1})$ and derives heterogeneous treatment regime effects. Moreover by taking the expectation, they could easily obtain average treatment regime effects. However in many settings including ours, this assumption is violated. To be more precise, users' treatment decisions are often driven by their likelihood to purchase the product.    

Such violation would make all the identification results (e.g., in \citet{ORR2010, P2016}) fail to be true. To address this issue, we propose the use of an instrumental variable regime in a similar sense to the instrumental variable in the LATE model. Specifically, assume that there are $(T + 1)$ instrumental variables $Z_j$'s for $j = 0, 1, ..., T$, one for each $W_j$. Let $\mathcal{Z}_j$ be the domain for $Z_j$ for each $j$. We assume further that each $\mathcal{Z}_j = \{0, 1\}$.

Now for each possible realization $z = (z_0, z_1, ..., z_T)$ of $Z = \overline{Z}_T$ and $w = (w_0, w_1, ..., w_T)$ of $W = \overline{W}_T$, we define the vectors of potential outcomes:
\begin{equation*}
	W(z) = (W_0(z_0), W_1(z_1), ..., W_T(z_T)) \; \text{ and }
\end{equation*}
\begin{equation*}
	O(z, w) = \left( X_0, Y_1(z_0, w_0), X_1(z_0, w_0), ..., Y_{T + 1}(\overline{z}_T, \overline{w}_T), X_{T + 1}(\overline{z}_T, \overline{w}_T) \right).
\end{equation*}  
We write
\begin{equation*}
	\mathcal{O} = \left\{\left.(W(z), O(z, w)) \; \right| \; z_j \in \mathcal{Z}_j \text{ and } w_j \in \mathcal{W}_j \text{ for } j = 0, 1, ..., T\right\}
\end{equation*}
to denote the set of all possible vectors of potential outcomes. More concretely, we make the following assumptions.

First, we implicitly assume that there is no interference between users: each user has his potential outcomes and his choice as a function of his past outcomes and instruments. Furthermore, a standard assumption is that the observed treatments and outcomes are consistent with the relevant potential treatments and outcomes.    

\begin{assumption}
\label{IVC}
\emph{(IVR Consistency)}
\begin{itemize}

	\item[(i)] For each $j \in \{0, 1, ..., T\}$, we have $W_j = W_j(Z_j)$.
	
	\item[(ii)] For each $j \in \{1, ..., T, T + 1\}$, we have $(Y_j, X_j) = \displaystyle \left(Y_j \left(\overline{Z}_{j - 1},  \overline{W}_{j - 1}\right), X_j \left(\overline{Z}_{j - 1},  \overline{W}_{j - 1}\right)\right)$. 

\end{itemize}
\end{assumption}

Next, we introduce a modified version of the sequential randomization assumption. We impose exclusion assumptions on $Z_j$'s to make them valid instruments. 

\begin{assumption}
\label{IVSR}
\emph{(IVR Sequential Randomization)}
\begin{itemize}
	\item[(i)] \emph{(Independence of IVR)} For each $j \in \{0, 1, ..., T \}$, $Z_j$ is conditionally independent of $\mathcal{O}$ given $\overline{O}_j$ and $\overline{Z}_{j - 1}$.
	
	\item[(ii)] \emph{(Exclusion of IVR)} 
	
\begin{itemize} 
    
    \item For each $j \in \{1, ..., T, T + 1\}$ and arbitrary $\overline{Z}_{j - 1}, \overline{Z^\prime}_{j - 1}$, and $\overline{W}_{j - 1}$, we have
\begin{equation*}
	\left(Y_j \left(\overline{Z}_{j - 1},  \overline{W}_{j - 1}\right), X_j \left(\overline{Z}_{j - 1},  \overline{W}_{j - 1}\right)\right) = \left(Y_j\left(\overline{Z^\prime}_{j - 1},  \overline{W}_{j - 1}\right), X_j\left(\overline{Z^\prime}_{j - 1},  \overline{W}_{j - 1}\right)\right).
\end{equation*}

    \item For each $j \in \{0, 1, ..., T\}$, 
    $$\{(Y_j(\overline{z}_{j - 1}, \overline{w}_{j - 1}), W_j(z_j)) | z_i, w_i \in \{0, 1\} \text{ for } i = 0, ..., j-1\} \perp Z_j.$$
\end{itemize}

\end{itemize}
\end{assumption}

Thanks to Assumption \ref{IVSR} (ii), we now can write $Y_j(\overline{W}_{j - 1})$ instead of $Y_j(\overline{Z}_{j - 1}, \overline{W}_{j - 1})$ and similarly for $X_j$. Also, $O(z, w) = O(w)$ for each $z$ and $w$. 

Before proceeding, we make another assumption that each $Z_j$ takes each value in its domain with positive probability. This assumption is essential as it necessitates the estimation of the causal effect of each treatment regime.

\begin{assumption}
\label{IVPO}
\emph{(IVR Positivity)} For each $j \in \{0, 1, ..., T\}$ and each realization of $\overline{O}_j$ and $\overline{Z}_j$, the following condition holds with probability $1$:
\begin{equation*}
	0 < P \left( Z_j = 1 | \overline{O}_j, \overline{Z}_{j - 1} \right) < 1.
\end{equation*}
\end{assumption}  

Recall that our main interest is to estimate the average causal effect of a treatment regime $(W_0, W_1, ..., W_T)$ on some measurable function $u(\cdot)$ of the outcomes $O$. Let $d$ denote a treatment regime that assigns each $W_j$ to a fixed value in $\mathcal{W}_j$. We write $W^d = (W^d_0, W^d_1, ..., W^d_T)$ to denote the treatment sequence if the subject had followed the regime $d$. Likewise, we write $O^d = O(W^d)$ to denote the vector of outcomes if the subject had followed the regime $d$. Then, the object of interest is $u(O^d)$. In many cases, $u(O^d) = Y_{T + 1}$ or $u(O^d) = \sum\limits_{t = 1}^{T+1} Y_t$. 

Now, we note that a realization $\overline{Z}_T$ of an IVR is a vector in $\{0, 1\}^{T + 1}$. For two realizations $\overline{Z}_T$ and $\overline{Z^\prime}_T$, we write
\begin{itemize}

    \item $\overline{W}_T(\overline{Z}_T) \succeq \overline{W}_T(\overline{Z^\prime}_T)$ if $W_j(Z_j) \geq W_j(Z^\prime_j)$ for all $j$ and there exists some $k$ such that $W_k(Z_k) > W_k(Z^\prime_k)$;
    
    \item $\overline{W}_T(\overline{Z}_T) \succ \overline{W}_T(\overline{Z^\prime}_T)$ if $W_j(Z_j) > W_j(Z^\prime_j)$ for all $j$.

\end{itemize}

Similarly to the Local Average Treatment Effect (LATE) model introduced by \citet{IA1994}, we are able to obtain identification results in this dynamic setting. 

With a single binary treatment variable $W$ and a corresponding binary instrument $Z$, traditional assumption is the monotonicity assumption, which states $P(W(1) \geq W(0) | X) = 1$. We make a similar assumption here. 

\begin{assumption}
\label{monotonicity} 
\emph{(Monotonicity)}
For each $j \in \{0, 1, ..., T\}$, we have 
\begin{equation*}
	P\left(W_j(1) \geq W_j(0)\right) = 1.
\end{equation*}
\end{assumption}

The last independence assumption we need is the independence between $W_j$'s in different periods conditional on common past information. 

\begin{assumption}
\label{indepacrosstime}
For any $j, k \in \{0, 1, ..., T\}$ with $j \ne k$, we have
$$\left. (W_j(0), W_j(1)) \perp (W_k(0), W_k(1))\; \right| \; \overline{O}_{\min(j, k)}, \overline{Z}_{\min(j, k) - 1}. $$
\end{assumption}

In the next section, we discuss the identification results for the two-period case. 

\section{Identification Results For Two-Period Model}
\label{illusModel}

We consider the case with $T = 1$. Denote $u(\cdot) = u(X_0,W_0,Y_1,X_1,W_1,Y_2)$. Denote $u_{ij}(\cdot)=u(X_0,i,,Y_1,X_1,j,Y_2)$. We are interested in the identification result of $\mathbb{E}\left[u(\cdot)|X_0,W_{j}(1)>W_{j}(0) \; \forall j = 0,1\right]$ and $\mathbb{E}\left[u_{ij}(\cdot)|X_0,W_{j}(1)>W_{j}(0) \; \forall j = 0,1\right]$. 
\subsection{\label{putthingstogether}Local Identification Results}
By DeMorgan's Law, we have
\begin{eqnarray}
    && P\big(W_{j}(1) > W_{j}(0) \; \forall j = 0, 1|X_0 \big) \times \mathbb{E}\left[u(\cdot)|X_0,W_{j}(1)>W_{j}(0) \; \forall j = 0,1\right] \nonumber\\
    &=& \mathbb{E}\left[u(\cdot)|X_0 \right] \nonumber \\
    && -\sum_{i=0,1} \mathbb{E} \left[u(\cdot)|X_0,W_0(1) = W_0(0)=i\right]P(W_0(1) = W_0(0)=i|X_0) \label{firstSum}\\
    && - \sum_{i=0,1} \mathbb{E}\big[u(\cdot)|X_0,W_0(1)>W_0(0),W_1(1) = W_1(0)=i \big] \nonumber \\
    && \hskip+2.5cm \times P(W_0(1)>W_0(0)|X_0)P \big(W_1(1) = W_1(0)=i|X_0\big). \label{secondSum}
\end{eqnarray}

Thus, in order to identify  $\mathbb{E}\left[u(\cdot)|X_0,W_{j}(1)>W_{j}(0) \; \forall j = 0,1\right]$ we need to identify three terms: Sum (\ref{firstSum}), Sum (\ref{secondSum}), and $P\big(W_{j}(1) > W_{j}(0) \; \forall j = 0, 1|X_0 \big)$; see appendix for derivation. Combining these identification results, we obtain
\begin{equation} \label{iden2period}             \mathbb{E}\left[u(\cdot)|X_0,W_{j}(1)>W_{j}(0) \right] = \frac{\mathbb{E}[\kappa u(\cdot)|X_0]}{P\big(W_{j}(1) > W_{j}(0) \; \forall j = 0, 1|X_0 \big)},
\end{equation}
where 
\begin{eqnarray*}
\kappa = &&1  - \frac{ W_0 (1-Z_0)}{P(Z_0=0|X_0,X_1)} - \frac{(1-W_0) Z_0}{P(Z_0=1|X_0,X_1)} - \frac{ W_1 (1-Z_1)}{P(Z_1=0|X_0,X_1)} - \frac{(1-W_1) Z_1}{P(Z_1=1|X_0,X_1)} \\
&& + \frac{W_0(1-Z_0)W_1(1-Z_1)}{P(Z_1=0,Z_0=0|X_0,X_1)}+ \frac{(1-W_0)Z_0W_1(1-Z_1)}{P(Z_1=0,Z_0=1|X_0,X_1)}\\
&& +\frac{W_0(1-Z_0)(1-W_1)Z_1}{P(Z_1=1,Z_0=0|X_0,X_1)}+ \frac{(1-W_0)Z_0(1-W_1)Z_1}{P(Z_1=1,Z_0=1|X_0,X_1)}
\end{eqnarray*}

\subsection{Local Average Treatment Regime Effect}
We can then express local average treatment regime effect in terms of Equation (\ref{iden2period}). For example, the local average effect of a regime switching $W_0$ from 0 to 1 and $W_1$ from 1 to 0, for a complier, for whom $W_j=Z_j$, is 
\begin{eqnarray}
    && \mathbb{E} \Big[u_{10}(\cdot) - u_{01}(\cdot)|X_0,W_j(1)>W_j(0) \; \forall j = 0, 1\Big]\nonumber\\
    &=& \frac{1}{P(W_j(1)>W_j(0) \; \forall j = 0, 1|X_0)} \Bigg[\mathbb{E} \Big[\kappa u(\cdot)\frac{W_0(1-W_1)}{P(Z_0=1,Z_1=0|X_0,X_1)}|X_0\Big] \nonumber \\
    && \hskip+6cm -\mathbb{E} \Big[\kappa u(\cdot)\frac{(1-W_0)W_1}{P(Z_0=0,Z_1=1|X_0,X_1)}|X_0\Big]\Bigg].
\end{eqnarray}
\subsection{Unconditional Local Identification}
As $X_0$ might typically has high dimensions, it is convenient to illustrate the local average treatment regime effect unconditionally. Now let's derive the expression for unconditional effect which we will use in simulation part. Applying Bayes' theorem and integrating yields
\begin{eqnarray}
&&\int \mathbb{E}\Big[u(\cdot|X_0,W_{j}(1)>W_{j}(0) )\Big]\text{d}P(X_0|W_{j}(1)>W_{j}(0) \; \forall j = 0, 1 ) \nonumber\\
&=&\int \mathbb{E}[\kappa u(\cdot)|X_0] \frac{\text{d}P(X_0|W_{j}(1)>W_{j}(0) \; \forall j = 0, 1)}{P\big(W_{j}(1) > W_{j}(0) \; \forall j = 0, 1|X_0 \big)} \nonumber\\
&=&\int \mathbb{E}[\kappa u(\cdot)|X_0] \frac{P(X_0,W_{j}(1)>W_{j}(0) \; \forall j = 0, 1)\text{d}X_0/P\big(W_{j}(1) > W_{j}(0) \; \forall j = 0, 1 \big)}{P\big(X_0,W_{j}(1) > W_{j}(0) \; \forall j = 0, 1\big)/P(X_0)} \nonumber\\
&=&\frac{1}{P\big(W_{j}(1) > W_{j}(0) \; \forall j = 0, 1 \big)} \int \mathbb{E}[\kappa u(\cdot)|X_0] \text{d}P(X_0).
\end{eqnarray}
Plugging Equation (3.4) into Equantion (3.5), we obtain the unconditional effect for a complier with regime switch $W_0$ from 0 to 1 and $W_1$ from 1 to 0
\begin{eqnarray}
&& \mathbb{E} \Big[u_{10}(\cdot) - u_{01}(\cdot)|W_j(1)>W_j(0) \; \forall j = 0, 1\Big]\nonumber\\
&=& \frac{\mathbb{E} [g(X_0)]}{P\big(W_{j}(1) > W_{j}(0) \; \forall j = 0, 1 \big)}, \label{idenEQ}  
\end{eqnarray}
where
\begin{equation*}
    g(X_0) = \mathbb{E} \Big[\kappa u(\cdot)\frac{W_0(1-W_1)}{P(Z_0=1,Z_1=0|X_0,X_1)}|X_0\Big] -\mathbb{E} \Big[\kappa u(\cdot)\frac{(1-W_0)W_1}{P(Z_0=0,Z_1=1|X_0,X_1)}|X_0\Big].
\end{equation*}
Next, we move on to the identification results for the general case. For the more general case, we only present the results for conditional effect; results for unconditional effect can be derived from the conditional version analogously.

\section{Main Identification Results}
\label{mainIden}

We denote $u(\cdot)$ = $u(X_0, W_0, Y_1, X_1, W_1, \cdots, X_T, W_T,Y_{T + 1})$. Let's start with the identification for local treatment effect on a full complier, for whom $W_j(1)>W_j(0)$ for all periods $j$. DeMorgan's Law gives
\begin{eqnarray}
    && P\Big(W_{j}(1) > W_{j}(0) \; \forall j|X_0 \Big) \times \mathbb{E}\Big[u(\cdot)|X_0,W_{j}(1)>W_{j}(0) \; \forall j\Big] \nonumber\\
    &=& \mathbb{E}\Big[u(\cdot)|X_0 \Big] \nonumber \\
    &&-\sum_{i \in \{0,1\}} \mathbb{E} \Big[u(\cdot)|X_0,W_0(1) = W_0(0)=i\Big] \times P\Big(W_0(1) = W_0(0)=i|X_0\Big) \nonumber\\
    &&-\sum_{\tau=0}^{T-1} \sum_{i \in \{0,1\}} \mathbb{E} \Big[u(\cdot)|X_0,W_j(1)>W_j(0), \forall j \leq \tau, W_{\tau+1}(1) = W_{\tau+1}(0)=i \Big] \nonumber\\ 
    && \hskip+2.5cm \times \left(\prod_{j=0}^{\tau}P(W_j(1)>W_j(0)|X_0)\right) \times P\Big(W_{\tau+1}(1) = W_{\tau+1}(0)=i|X_0\Big). \nonumber
\end{eqnarray}
Therefore, in order to identify $\mathbb{E}\Big[u(\cdot)|X_0,W_{j}(1)>W_{j}(0) \; \forall j\Big]$ we need to identify the sums on the right hand side of the above equality and the conditional probabilities $P\Big(W_{j}(1) > W_{j}(0)|X_0 \Big)$.

\subsection{\label{firststepIden}First Step Identification}
We denote $K_{t,0} = W_t(1-Z_t)$ and $K_{t,1} = (1-W_t)Z_t$. We postpone the identification of conditional probabilities for a moment. We apply DeMorgan's Law to evaluate the expected utility of a full-complier:
\begin{eqnarray}
    &&\mathbb{E}\Big[u(\cdot)|X_0,W_{j}(1)>W_{j}(0) \; \forall j\Big] \nonumber\\
    &=& \frac{\mathbb{E}[\kappa u(\cdot)|X_0]}{P\big(W_{j}(1) > W_{j}(0) \; \forall j|X_0 \big)},
\end{eqnarray}
where
\begin{eqnarray*}
\kappa = 1+\sum_{\tau=1}^{T} (-1)^\tau \sum_{\substack{i_1,\cdots,i_\tau \\ \in \{0,1\}}}\sum_{\substack{j_1<\cdots<j_\tau \\ \in \{0,1,\cdots, T\}}}\frac{\prod \limits_{t=1}^\tau K_{j_t,i_t}}{P(Z_{j_1}=i_1,\cdots,Z_{j_\tau} = i_\tau|X_0,\cdots,X_{j_\tau})}.
\end{eqnarray*}
In one treatment case, we cannot learn anything about non-compliers; in contrast, we can learn about local average treatment effects for people of different compliance types. Indeed, program researchers may be more interested in treatment effects on different compliance types other than the full-compliance. To this end, denote the periods of compliance by
\begin{equation*} 
    \mathcal{T}_c:=\{j\in\{0,\cdots, T\}|W_{j}(1)>W_{j}(0)\},
\end{equation*}
and the periods of non-compliance by 
\begin{equation*}
    \mathcal{T}^0_n:=\{j\in\{0,\cdots, T\}|W_{j}(1)=W_{j}(0)=0\}, \text{ and }
\end{equation*} 
\begin{equation*}
    \mathcal{T}^1_n:=\{j\in\{0,\cdots, T\}|W_{j}(1)=W_{j}(0)=1\}.
\end{equation*}
Note that under the monotonicity assumption \ref{monotonicity}, we can consider only the compliance type. 
Then a compliance type can be represented by a tuple $(\mathcal{T}_c, \mathcal{T}^0_n, \mathcal{T}^1_n)$. For instance, the full-compliance type is represented by the tuple $(\{0,\cdots,T\},\phi,\phi)$. We define a factor $\kappa$ associated with compliance type $(\mathcal{T}_c, \mathcal{T}^0_n, \mathcal{T}^1_n)$ by
\begin{eqnarray}
\kappa &=& \frac{\prod_{t\in \mathcal{T}^0_n}K_{t,0}\prod_{t\in \mathcal{T}^1_n}K_{t,1}}{P\big(Z_j = i \forall j \in \mathcal{T}^i_n,i=0,1|X_0\big)}\times\nonumber\\
    &&\left(1+\sum_{\tau=1}^{|\mathcal{T}_c|} (-1)^\tau \sum_{\substack{i_1,\cdots,i_\tau \\ \in \{0,1\}}}\sum_{\substack{j_1<\cdots<j_\tau \\ \in \mathcal{T}_c}}\frac{\prod \limits_{t=1}^\tau K_{j_t,i_t}}{P(Z_{j_1}=i_1,\cdots,Z_{j_\tau} = i_\tau|X_0,\cdots,X_{j_\tau})}\right),\nonumber    
\end{eqnarray}
where $\kappa = \kappa({\mathcal{T}_c, \mathcal{T}^0_n, \mathcal{T}^1_n},X_0)$ is a function of compliance type and period-0 covariates. When there is no ambiguity, we leave out the type and covariate arguments for $\kappa$. As a general result, the expected utility of an agent with compliance type $(\mathcal{T}_c, \mathcal{T}^0_n, \mathcal{T}^1_n)$ is
\begin{equation}
\mathbb{E}\left[u(\cdot)|X_0,(\mathcal{T}_c,\mathcal{T}^0_n,\mathcal{T}^1_n)\right] = \frac{\mathbb{E}\left[\kappa~ u(\cdot) | X_0 \right]}{P\big(\mathcal{T}_c,\mathcal{T}^0_n,\mathcal{T}^1_n|X_0 \big)}.
\end{equation}
The identification results will be complete once we can identify
\begin{eqnarray}
&& P\big(\mathcal{T}_c,\mathcal{T}^0_n,\mathcal{T}^1_n|X_0 \big) \nonumber \\ 
    &\equiv& P \Big(W_j(1) > W_j(0) \; \forall j\in \mathcal{T}_c \; \& \; W_j(1) = W_j(0) = 0 \; \forall j \in \mathcal{T}_n^0 \; \& \; W_j(1) = W_j(0) = 1 \; \forall j \in \mathcal{T}_n^1 | X_0\Big), \nonumber
\end{eqnarray} 
which is the focus of the next section.

\subsection{\label{fullprobEst}Main Probability Estimation}
Let us first consider the full-compliance case, that is, $\mathcal{T}_c = \{0, ..., T\}$, keeping in mind that the general case is similar. In other words, we consider 
\begin{equation*}
    P\big(\mathcal{T}_c,\mathcal{T}^0_n,\mathcal{T}^1_n|X_0 \big) = P(W_j(1) > W_j(0) \; \forall j | X_0).
\end{equation*}
Now under Assumption \ref{indepacrosstime}, we have (see the iterative derivation in Appendix)
\begin{equation}
 P(W_j(1) > W_j(0) \; \forall j | X_0) = \prod_{j = 0}^T P\big(W_j(1) > W_j(0) | X_0\big).
\end{equation}
Similarly, for any set $(i_0, ..., i_T) \in \{0, 1\}^{T + 1}$ we have
$$ \mathbb{E}\left[\left. \prod_{j = 0}^T W_j(i_j) \right| X_0\right] = P\big(W_j(i_j) = 1 \; \forall \, j | X_0\big) = \prod_{j = 0}^T P\big(W_j(i_j) = 1 | X_0\big) = \prod_{j = 0}^T \mathbb{E}\big[W_j(i_j) | X_0\big].$$
On the other hand by Assumption \ref{monotonicity},
\begin{eqnarray}
P(W_j(1) > W_j(0) | X_0)
&=& 1 - P\big(W_j(0) = 1 | X_0) - P(W_j(1) = 0 | X_0\big) \nonumber \\
&=& P\big(W_j(1) = 1 | X_0) - P(W_j(0) = 1 | X_0\big) \nonumber \\
&=& \mathbb{E}\big[W_j(1) | X_0\big] - \mathbb{E}\big[W_j(0) | X_0\big]. \nonumber
\end{eqnarray}
Therefore, we can rewrite $P\big(W_j(1) > W_j(0\big) \; \forall \, j | X_0)$ as
\begin{eqnarray}
\prod_{j = 0}^T \bigg( \mathbb{E}\big[W_j(1) | X_0\big] - \mathbb{E}\big[W_j(0) | X_0\big] \bigg)
= \sum_{(i_0, ..., i_T) \in \{0, 1\}^{T + 1}} (-1)^{T + 1 - \sum_{j = 0}^T i_j} \cdot \mathbb{E}\left[\left. \prod_{j = 0}^T W_j(i_j) \right| X_0\right]. \nonumber
\end{eqnarray} 
Thus, we can determine the probability $P(W_j(1) = 1, W_j(0) = 0 \; \forall \, j | X_0)$ if each term $\mathbb{E}\left[\left. \prod_{j = 0}^T W_j(i_j) \right| X_0\right]$ is determined. This can be done according to the following lemma.  
\begin{lemma}
\label{determineProb}
Fix $(i_0, ..., i_T) \in \{0, 1\}^{T + 1}$. Then under Assumptions \ref{IVC}, \ref{IVSR}, \ref{IVPO}, we have
$$ \mathbb{E} \left[\left. \prod_{j = 0}^T W_j \cdot \prod_{j = 0}^T \frac{\textbf{1}_{\{Z_j = i_j\}}}{P(Z_j = i_j | \overline{O}_j, \overline{Z}_{j - 1})} \right| X_0\right] = \mathbb{E}\left[\left. \prod_{j = 0}^T W_j(i_j) \right| X_0\right]. $$
\end{lemma}
\begin{proof}
Lemma \ref{determineProb} is indeed a special case of Theorem $3.1$ in \cite{P2016}. 
\end{proof}

Lemma \ref{determineProb} and the above reasoning implies Theorem \ref{identificationProb} below, which is the identification result for $P(W_j(1) > W_j(0) \; \forall \, j | X_0)$.
\begin{theorem}
\label{identificationProb}
Under Assumptions \ref{IVC}, \ref{IVSR}, \ref{IVPO}, and \ref{indepacrosstime}, we have
$$P(W_j(1) > W_j(0) \; \forall \, j | X_0) = \mathbb{E} \left[\left. \prod_{j = 0}^T W_j \cdot \sum_{(i_0, ..., i_T) \in \{0, 1\}^{T + 1}} \prod_{j = 0}^T \frac{ (-1)^{1 - i_j} \cdot \textbf{1}_{\{Z_j = i_j\}}}{P(Z_j = i_j | \overline{O}_j, \overline{Z}_{j - 1})} \right| X_0\right].$$ 
\end{theorem}

In a similar manner, we can fully identify $P(\mathcal{T}_c, \mathcal{T}_n^0, \mathcal{T}_n^1 | X_0)$ in the general case. See proof in the appendix.
\begin{theorem}
Under Assumptions \ref{IVC}, \ref{IVSR}, \ref{IVPO}, and \ref{indepacrosstime}, we have
\begin{eqnarray*}
P(\mathcal{T}_c, \mathcal{T}_n^0, \mathcal{T}_n^1 | X_0) = \mathbb{E} \left[\prod_{j \in \mathcal{T}_c\cup\mathcal{T}_n^1} W_j \cdot \prod_{j \in \mathcal{T}_n^0} (1-W_j) \cdot \Bigg(\sum_{(i_j)_{j\in \mathcal{T}_c} \in \{0, 1\}^{|\mathcal{T}_c|}} \prod_{j \in \mathcal{T}_c} \frac{ (-1)^{1 - i_j} \cdot \textbf{1}_{\{Z_j = i_j\}}}{P(Z_j = i_j | \overline{O}_j, \overline{Z}_{j - 1})} ~\Bigg) \right.\\
\left.\left.\cdot \prod_{j \in \mathcal{T}_n^0} \frac{\textbf{1}_{\{Z_j = 1\}}}{P(Z_j = 1| \overline{O}_j, \overline{Z}_{j - 1})}  \prod_{j \in \mathcal{T}_n^1} \frac{\textbf{1}_{\{Z_j = 0\}}}{P(Z_j = 0| \overline{O}_j, \overline{Z}_{j - 1})} \right| X_0\right].
\end{eqnarray*}
\end{theorem}

\section{\label{simu}Application}
To identify or apply our above results to estimate local treatment effects, one needs to know the compliance type of each observation in the sample. Here are a few scenarios that are ideal for application: if the researcher has good institutional knowledge that helps him classify the compliance type based on each subject's covariates; if the whole sample consists of a single compliance type so that one can easily identify the compliance type based on the realized treatment regimes. The following simulation exercise shows how to apply our results to a full compliance environment.  

\subsection{Simulation Setup}
For this section, we consider a two-period setting with $T = 1$. The outcome of interest is $u(\cdot) = Y_{T+1} \equiv Y_2$. 

The simulation procedure consists of repeating the following process $500$ times.

We generate $n = 500,000$ observations of the following variables. 
\begin{enumerate}

    \item Simulate $X_0 \in U[-1, 1]^{6}$.
    
    \item Set $X_1 \equiv X_0$. 
    
    \item Simulate $\epsilon_0 \sim U[0, 1]$ and $Z_0 \in \{0, 1\}$ with $\displaystyle P(Z_0 = 1 | X_0) = 1 / (1 + \exp(X_0 \xi))$. 
    
    \item Simulate $W_0 \in \{0, 1\}$ such that
\begin{equation*}
    P(W_0 = 1) = 
    \begin{cases} 
        \epsilon_0 & \mbox{ if } Z_0 = 0 \\
        1 & \mbox{ if } Z_0 = 1
    \end{cases}.
\end{equation*}

    \item Generate 
\begin{equation*}
    Y_1 = X_0 \alpha_1 + \beta_1 W_0 + \epsilon_0.
\end{equation*}
    
    \item Simulate $\epsilon_1 \sim U[0, 1]$ and $Z_1 \in \{0, 1\}$ with $P(Z_1 = 1) = e_1 \in (0, 1)$. 
    
    \item Simulate $W_1 \in \{0, 1\}$ such that
\begin{equation*}
    P(W_1 = 1) = 
    \begin{cases} 
        \epsilon_1 & \mbox{ if } Z_1 = 0 \\
        1 & \mbox{ if } Z_1 = 1
    \end{cases}.
\end{equation*}

    \item Generate 
\begin{equation*}
    Y_2 = X_0 \alpha_2 + \beta_2 W_0 + \delta Y_1 + \gamma W_1 + \epsilon_1.
\end{equation*}
So
\begin{equation*}
    Y_2 = X_0 (\alpha_2 + \delta \alpha_1) + (\beta_2 + \delta \beta_1) W_0 + \gamma W_1 + \delta \epsilon_0 + \epsilon_1
\end{equation*}

\end{enumerate}
Here, the parameters are chosen to be
\begin{itemize}
    \item $\xi = [1, 2, 3, -1, -2, -3]^T$; $e_1 = 0.75$
    \item $\alpha_1 = [1, 1, 1, 1, 1, 2, 2, 2, 2, 2]^T$; $\beta_1 = 2$.
    \item $\alpha_2 = [2, 2, 2, 2, 2, 1, 1, 1, 1, 1]^T$; $\beta_2 = 2$; $\delta = 2$; $\gamma = 1$. 
\end{itemize}

\subsection{Simulation Result}
We want to estimate the local effect of regime (1, 0) with respect to regime (0, 1). Under our assumptions, the true conditional (local) average treatment regime effect is $\tau = (\beta_2 + \delta \beta_1) - \gamma = 5$.  \\

We compare our proposed method with two other methods:
\begin{itemize}
    \item Naive approach: We take the difference in means of two sets of observations corresponding to regime (1, 0) and regime (0, 1).
    
    \item No IVs: We assume sequential treatments, but do not use IVs. The formula will be similar to Equation (\ref{idenEQ}) except that in the numerator, we use the formula in \citet{P2016} with only $W_0, W_1$ instead of $Z_0, Z_1$. 
\end{itemize}

We compare these methods using the $500$ simulated datasets generated above in terms of four measures:

\begin{itemize}

	\item absolute mean error: $|\mathbb{E}[\widehat{\tau} - \tau]|$.
	
	\item mean absolute error: $\mathbb{E}[\lvert\widehat{\tau} - \tau\rvert]$. 

	\item absolute median error: $|median(\widehat{\tau} - \tau)|$.
	
	\item median absolute error: $median(\lvert\widehat{\tau} - \tau\rvert)$.

\end{itemize}
Here, $\tau$ is the true outcome ($\tau = 5$) while $\widehat{\tau}$ is an estimator of $\tau$. The results are summarized in Table \ref{simRes}.

\begin{table}[ht]
\centering
\begin{tabular}{@{}l | cccc@{}}
  \hline
\multirow{2}{*}{Error Metric} & Absolute & Mean & Absolute & Median \\
             & Mean     & Absolute & Median & Absolute \\
  \hline
\texttt{Naive} & 0.86 & 0.86 & 0.86 & 0.86 \\ 
\texttt{No IV} & 1.24 & 1.24 & 1.24 & 1.24 \\ 
\textbf{\texttt{LATRE}} & \textbf{0.56} & \textbf{0.56} & \textbf{0.54} & \textbf{0.54} \\ 
   \hline
\end{tabular}
\caption{\label{simRes}Simulation Results. (Smaller is better.)} 
\end{table}

As we can see, if we do not use IVR when the problem of endogenous selections of sequential treatments is present then the estimate is terrible. In this simulation study, it is even worse than the naive estimate. The use of IVR (that is, our method LATRE) makes the estimation significantly better, which outperforms the naive estimate.  

\section{Conclusions}
\label{concl}

Following \citet{MVR2001, ORR2010, P2016} among others, we attack the setting with a treatment sequence. Generalizing the LATE model of \citet{IA1994}, we provide a method of using IVR to estimate the local average treatment regime effects. There are many research studies that use IVs to study causal effects, regarding the LATE model (notably \citet{IA1994, AIR1996, A2003}) as well as in the dynamic setting (see \citet{HL2004, O2012, SHBT2012, W2002}). However, all of them require a structural functional form assumption on the potential outcomes. In treatment regime settings, this assumption is unlikely to hold which makes their approaches suffer from the misspecification problem. 

On the one hand, we provide a theoretical framework for estimating local average treatment regime effects, which is robust to model misspecification since it does not require any structural assumption. On the other hand, we demonstrate the proposed method's performance via a simulated dataset in a two-treatment setting; the method would certainly be useful with a longer sequence of treatments as well.

\begin{appendices}

\section{Derivation of Equation (3.3)}
\subsection{\label{sum1Est}First Sum Identification}

By Bayes rule we have
\begin{eqnarray}
    && \mathbb{E} \left[W_0 (1-Z_0) u(\cdot)|X_0\right] \nonumber\\
    &=& \mathbb{E} \left[u(\cdot)|X_0,W_0=1,Z_0=0\right]P(W_0=1|X_0,Z_0 = 0)P(Z_0=0|X_0) ;\nonumber\\
    && \mathbb{E} \left[(1-W_0) Z_0 u(\cdot)|X_0\right] \nonumber\\
    &=& \mathbb{E} \left[u(\cdot)|X_0,W_0=0,Z_0=1\right]P(W_0=0|X_0,Z_0 = 1)P(Z_0=1|X_0). \nonumber
\end{eqnarray}
By independence assumption and the above lemma, we can rewrite the summation (\ref{firstSum}) as 
\begin{eqnarray}
\frac{\mathbb{E} \left[W_0 (1-Z_0)u(\cdot)|X_0\right]}{P(Z_0=0|X_0)}+
\frac{\mathbb{E} \left[(1-W_0) Z_0 u(\cdot)|X_0\right]}{P(Z_0=1|X_0)}.\nonumber
\end{eqnarray}

\subsection{\label{sum2Est}Second Sum Identification}
Similarly we have 
\begin{eqnarray}
    && \mathbb{E} \left[u(\cdot)|X_0,X_1,W_1(1) = W_1(0)=1 \right] P(W_1(1) = W_1(0)=1|X_0,X_1)\nonumber\\
    &=& \frac{\mathbb{E} \left[W_1 (1-Z_1)u(\cdot)|X_0,X_1\right]}{P(Z_1=0|X_0,X_1)};\nonumber\\
    && \mathbb{E} \left[u(\cdot)|X_0,X_1,W_1(1) = W_1(0)=0 \right] P(W_1(1) = W_1(0)=0|X_0,X_1)\nonumber\\
    &=& \frac{\mathbb{E} \left[(1-W_1) Z_1u(\cdot)|X_0,X_1\right]}{P(Z_1=1|X_0,X_1)}.\nonumber
\end{eqnarray}
Note that we can express conditional analogs of terms in summation (\ref{secondSum}) as
\begin{eqnarray}
    && \mathbb{E} \left[u(\cdot)|X_0,X_1, W_0(1)>W_0(0),W_1(1) = W_1(0)=i \right] \times P(W_0(1)>W_0(0)|X_0,X_1)\nonumber\\
    &=& \mathbb{E} \left[u(\cdot)|X_0,X_1,W_1(1) = W_1(0)=i\right]\nonumber\\
    && - \sum_j \mathbb{E} \left[u(\cdot)|X_0,X_1,W_0(1)=W_0(0)=j, W_1(1) = W_1(0)=i\right] \nonumber \\
    && \hskip+1cm \times P(W_0(1) = W_0(0)=j|X_0,X_1). \nonumber
\end{eqnarray}

By Bayes rule we have
\begin{eqnarray}
    && \mathbb{E}[W_0(1-Z_0)W_1(1-Z_1)u(\cdot)|X_0,X_1] \nonumber \\
    &=& \mathbb{E}\left[u(\cdot)|X_0,X_1,W_0=W_1=1,Z_0=Z_1=0\right] \nonumber \\
    && \hskip+0.5cm \times P(W_0=W_1=1|X_0,X_1,Z_0 = Z_1=0)P(Z_0=Z_1=0|X_0,X_1). \nonumber
\end{eqnarray}
Then using the independence assumption, we can rewrite the summation (\ref{secondSum}) as
\begin{eqnarray}
    && \mathbb{E}_{X_1|X_0} \bigg[\frac{E\left[W_1 (1-Z_1)u(\cdot)|X_0,X_1\right]}{P(Z_1=0|X_0,X_1)}+
    \frac{\mathbb{E}\left[(1-W_1) Z_1 u(\cdot)|X_0,X_1\right]}{P(Z_1=1|X_0,X_1)} \nonumber\\
    && - \frac{\mathbb{E}[W_0(1-Z_0)W_1(1-Z_1)u(\cdot)|X_0,X_1]}{P(Z_1=0,Z_0=0|X_0,X_1)}- \frac{\mathbb{E}[(1-W_0)Z_0W_1(1-Z_1)u(\cdot)|X_0,X_1]}{P(Z_1=0,Z_0=1|X_0,X_1)}\nonumber\\
    && - \frac{\mathbb{E}[W_0(1-Z_0)(1-W_1)Z_1u(\cdot)|X_0,X_1]}{P(Z_1=1,Z_0=0|X_0,X_1)}- \frac{\mathbb{E}[(1-W_0)Z_0(1-W_1)Z_1u(\cdot)|X_0,X_1]}{P(Z_1=1,Z_0=1|X_0,X_1)}\bigg].\nonumber
\end{eqnarray}

\subsection{\label{probEst}Probability Identification}
Under Assumption \ref{indepacrosstime}, we have
\begin{eqnarray}
    && P\Big(W_j(1) > W_j(0) \; \forall j = 0, 1 | X_0\Big) \nonumber \\
    &=& P\Big(W_0(1) > W_0(0) | X_0\Big) \times P\Big(W_1(1) > W_1(0) | W_0(1) > W_0(0), X_0\Big) \nonumber \\
    &=& P\Big(W_0(1) > W_0(0) | X_0\Big) \times P\Big(W_1(1) > W_1(0) | X_0\Big).  \nonumber
\end{eqnarray}

Similarly, for any set $(i_0, i_1) \in \{0, 1\}^2$ we have
\begin{eqnarray}
\mathbb{E}\Big[W_0(i_0) \times W_1(i_1) | X_0 \Big]
&=& P\Big(W_j(i_j) = 1 \; \forall \, j = 0, 1 | X_0\Big) \nonumber \\
&=& P\Big(W_0(i_0) = 1 | X_0\Big) \times P\Big(W_1(i_1) = 1 | X_0\Big) \nonumber \\
&=& \mathbb{E}\Big[W_0(i_0) | X_0\Big] \times \mathbb{E}\Big[W_1(i_1) | X_0\Big]. \nonumber  
\end{eqnarray}

On the other hand by Assumption \ref{monotonicity},
\begin{eqnarray}
P(W_j(1) > W_j(0) | X_0)
&=& 1 - P\big(W_j(0) = 1 | X_0) - P(W_j(1) = 0 | X_0\big) \nonumber \\
&=& P\big(W_j(1) = 1 | X_0) - P(W_j(0) = 1 | X_0\big) \nonumber \\
&=& \mathbb{E}\big[W_j(1) | X_0\big] - \mathbb{E}\big[W_j(0) | X_0\big]. \nonumber
\end{eqnarray}

Therefore,
\begin{eqnarray}
&& P\big(W_j(1) > W_j(0\big) \; \forall \, j = 0, 1 | X_0) \nonumber \\
&= & \Big(\mathbb{E}\big[W_0(1) | X_0\big] - \mathbb{E}\big[W_0(0) | X_0\big]\bigg) \times \Big(\mathbb{E}\big[W_1(1) | X_0\Big] - \mathbb{E}\big[W_1(0) | X_0\big]\Big) \nonumber \\
&=& \mathbb{E}\Big[ W_0(1) W_1(1) - W_0(0) W_1(1) - W_0(1) W_1(0) + W_0(0) W_1(0) | X_0\Big]. \nonumber
\end{eqnarray} 

Thus, we can determine the probability $P(W_j(1) = 1, W_j(0) = 0 \; \forall \, j | X_0)$ if each term $\mathbb{E}\left[W_0(i_0) W_1(i_1) | X_0\right]$ is determined.

To this end, a direct application of Theorem $3.1$ in \cite{P2016} implies 

\begin{equation*}
    \mathbb{E} \Big[W_0(1) W_1(1) | X_0\Big] =
    \mathbb{E} \Bigg[ W_0 W_1 \times \frac{Z_0 Z_1}{P(Z_0 = 1 | X_0) P(Z_1 = 1 | X_0, Z_0, X_1, Y_1)} \Bigg| X_0 \Bigg].
\end{equation*}
Similarly, $P\big(W_j(1) > W_j(0\big) \; \forall \, j = 0, 1 | X_0)$ is identified.  

\section{Derivation of Equation (3.4)}
\begin{eqnarray}
    && \mathbb{E} \Big[u_{10}(\cdot) - u_{01}(\cdot)|X_0,W_j(1)>W_j(0) \; \forall j = 0, 1\Big]\nonumber\\
    &=& \mathbb{E} \Big[u(\cdot)|X_0,W_0=1,W_1=0,W_j(1)>W_j(0) \; \forall j = 0, 1\Big]\nonumber \\
    && \hskip+2.5cm -\mathbb{E}\Big[u(\cdot)|X_0,W_0=0,W_1=1,W_j(1)>W_j(0) \; \forall j = 0, 1\Big] \nonumber\\
    &=& \mathbb{E} \Big[u(\cdot)|X_0,Z_0=1,Z_1=0,W_j(1)>W_j(0) \; \forall j = 0, 1\Big] \nonumber \\
    && \hskip+2.5cm -\mathbb{E}\Big[u(\cdot)|X_0,Z_0=0,Z_1=1,W_j(1)>W_j(0) \; \forall j = 0, 1\Big] \nonumber\\
    &=& \frac{1}{P(W_j(1)>W_j(0) \; \forall j = 0, 1|X_0)} \Bigg[\mathbb{E} \Big[\kappa u(\cdot)\frac{W_0(1-W_1)}{P(Z_0=1,Z_1=0|X_0,X_1)}|X_0\Big] \nonumber \\
    && \hskip+6cm -\mathbb{E} \Big[\kappa u(\cdot)\frac{(1-W_0)W_1}{P(Z_0=0,Z_1=1|X_0,X_1)}|X_0\Big]\Bigg].
\end{eqnarray}

\section{Derivation of Equation (4.1)}
First DeMorgan's Law gives 
\begin{eqnarray}
&& \mathbb{E}\Big[u(\cdot)|X_0,W_j(1)>W_j(0), \forall j \leq \tau, W_{\tau+1}(1) = W_{\tau+1}(0)=i \Big] \nonumber\\ 
&& \hskip+2cm \times \left(\prod_{j=0}^{\tau}P(W_j(1)>W_j(0)|X_0)\right) - \mathbb{E}\Big[u(\cdot)|X_0, W_{\tau+1}(1) = W_{\tau+1}(0)=i\Big] \nonumber \\
&=& \sum_{k=1}^\tau (-1)^k \sum_{\substack{i_1,\cdots,i_k \\ \in \{0,1\}}}\sum_{\substack{j_1<\cdots<j_k \\ \in\{0,1,\cdots,\tau\}}} \mathbb{E}\Big[u(\cdot)|X_0,W_{j_m}(1) = W_{j_m}(0)=i_m \; \forall m \in \{1, ..., k\}, \nonumber \\
&& \hskip+10cm W_{\tau+1}(1) = W_{\tau+1}(0)=i\Big].  \nonumber
\end{eqnarray}
Thus the expected utility of a full-complier is
\begin{eqnarray}
    &&\mathbb{E}\Big[u(\cdot)|X_0,W_{j}(1)>W_{j}(0) \; \forall j\Big] \nonumber\\
    &=& \frac{1}{P\big(W_{j}(1) > W_{j}(0) \; \forall j|X_0 \big)} \times \Bigg(\mathbb{E}\Big[u(\cdot)|X_0\Big] + \nonumber\\
    && \sum_{\tau=1}^T (-1)^\tau \sum_{\substack{i_1,\cdots,i_\tau \\ \in \{0,1\}}}\sum_{\substack{j_1<\cdots<j_\tau \\ \in\{0,1,\cdots,T\}}}E_{X_1,\cdots,X_{j_\tau}|X_0}\left[\frac{\mathbb{E} \Big[u(\cdot)\prod\limits_{t=1}^\tau K_{j_t,i_t}|X_0,\cdots,X_{j_\tau}\Big]}{P(Z_{j_1}=i_1,\cdots,Z_{j_\tau} = i_\tau|X_0,\cdots,X_{j_\tau})}\right]\Bigg)\nonumber\\
    &=& \frac{\mathbb{E}[\kappa u(\cdot)|X_0]}{P\big(W_{j}(1) > W_{j}(0) \; \forall j|X_0 \big)},\nonumber
\end{eqnarray}

\section{Derivation of Equation (4.3)}
By Assumption 2.5, we have
\begin{eqnarray}
    && P(W_j(1) > W_j(0) \forall j | X_0) \nonumber \\
    &=& P(W_0(1) > W_0(0) | X_0) \cdot P(W_j(1) > W_j(0) \forall j \geq 1 | W_0(1) > W_0(0), X_0) \nonumber \\
    &=& P(W_0(1) > W_0(0) | X_0) \cdot P(W_j(1) > W_j(0) \forall j \geq 1 | X_0)  \nonumber \\
    &=& P(W_0(1) > W_0(0) | X_0) \cdot P(W_1(1) > W_1(0) | X_0) \nonumber \\ 
    && \hskip+1.5cm \times \; \mathbb{E}_{\overline{O}_1, Z_0 | X_0} \big[ P(W_j(1) > W_j(0) \forall j \geq 2 | W_1(1) > W_1(0), \overline{O}_1, Z_0) \big]  \nonumber \\
    &=& P(W_0(1) > W_0(0) | X_0) \cdot P(W_1(1) > W_1(0) | X_0) \nonumber \\ 
    && \hskip+1.5cm \times \; \mathbb{E}_{\overline{O}_1, Z_0 | X_0} \big[ P(W_j(1) > W_j(0) \forall j \geq 2 | \overline{O}_1, Z_0) \big]  \nonumber \\
    &=& P(W_0(1) > W_0(0) | X_0) \cdot P(W_1(1) > W_1(0) | X_0) \cdot P(W_j(1) > W_j(0) \forall j \geq 2 | X_0). \nonumber
\end{eqnarray}
Continuing this process, we obtain the result.

\section{Proof of Theorem (4.3)}
First note that under Assumption \ref{indepacrosstime}, we have can expand $P\big(\mathcal{T}_c,\mathcal{T}^0_n,\mathcal{T}^1_n|X_0 \big)$ as a product form
\begin{eqnarray*}
 \prod_{j\in \mathcal{T}_c} P \Big(W_j(1) > W_j(0)|X_0\Big) \prod_{j \in \mathcal{T}_n^0} P \Big(W_j(1) = W_j(0) = 0|X_0\Big) \prod_{j \in \mathcal{T}_n^1} P \Big(W_j(1) = W_j(0) = 1|X_0\Big). 
\end{eqnarray*} 
Next by Assumption \ref{monotonicity}, 
\begin{eqnarray}
P(W_j(1) = W_j(0) = 0 | X_0)
&=& P(W_j(1) = 0 | X_0) = \mathbb{E}\big[1-W_j(1) | X_0\big];\nonumber \\
P(W_j(1) = W_j(0) = 1 | X_0)
&=& P(W_j(0) = 1 | X_0) = \mathbb{E}\big[W_j(0) | X_0\big]. \nonumber
\end{eqnarray}
We can rewrite $P\big(\mathcal{T}_c,\mathcal{T}^0_n,\mathcal{T}^1_n|X_0 \big)$ as
\begin{eqnarray}
&&  \prod_{j\in \mathcal{T}_c} \bigg( \mathbb{E}\big[W_j(1) | X_0\big] - \mathbb{E}\big[W_j(0) | X_0\big] \bigg) \prod_{j \in \mathcal{T}_n^0} \bigg( \mathbb{E}\big[1-W_j(1) | X_0\big]\bigg) \prod_{j \in \mathcal{T}_n^1} \bigg(\mathbb{E}\big[W_j(0) | X_0\big] \bigg) \nonumber \\
&=& \bigg(\sum_{(i_j)_{j\in \mathcal{T}_c} \in \{0, 1\}^{|\mathcal{T}_c|}} (-1)^{|\mathcal{T}_c| - \sum_{j\in \mathcal{T}_c} i_j} \cdot\mathbb{E}\left[\left. \prod_{j\in \mathcal{T}_c} W_j(i_j) \right| X_0\right]\bigg)\nonumber\\
&& \hskip+2.5cm\cdot\mathbb{E}\left[\left. \prod_{j\in \mathcal{T}_n^0} (1-W_j(1)) \right| X_0\right]
\mathbb{E}\left[\left. \prod_{j\in \mathcal{T}_n^1} (W_j(0)) \right| X_0\right]. \nonumber
\end{eqnarray} 
Applying Lemma \ref{determineProb}, we obtain the identification result for general compliance type $P(\mathcal{T}_c, \mathcal{T}_n^0, \mathcal{T}_n^1 | X_0)$.

\end{appendices}

\bibliographystyle{Chicago}
\bibliography{JASA-local_ave_treat_reg}
\end{document}